\documentclass[11pt]{article}

\usepackage{amsmath,amssymb,amsfonts,amsthm,epsfig}
\usepackage[usenames,dvipsnames]{xcolor}
\usepackage{bm,xspace}
\usepackage{tcolorbox}
\usepackage{cancel}
\usepackage{fullpage}
\usepackage{liyang}
\usepackage{framed}
\usepackage{verbatim}
\usepackage{enumitem}
\usepackage{array}
\usepackage{multirow}
\usepackage{afterpage}
\usepackage{mathrsfs}
\usepackage{pifont} 
\usepackage{chngpage}
\usepackage[normalem]{ulem}
\usepackage{boxedminipage}
\usepackage{caption}

\def\colorful{0}
\newcommand{\purple}[1]{#1}

\ifnum\colorful=1
\newcommand{\violet}[1]{{\color{violet}{#1}}}

\fi
\ifnum\colorful=0
\newcommand{\violet}[1]{{{#1}}}

\fi

\newcommand{\cost}{\mathrm{cost}}
 
\newcommand{\bias}{\mathrm{error}_{\pm 1}} 
\newcommand{\error}{\mathrm{error}}
\newcommand{\query}{\mathrm{query}}
\newcommand{\paths}{\mathrm{path}}
\newcommand{\leaf}{\mathrm{leaf}}
\newcommand{\nexti}{\mathrm{next\_i}}

\newcommand{\pparagraph}[1]{\bigskip \noindent {\bf {#1}}}

\makeatletter
\newtheorem*{rep@theorem}{\rep@title}
\newcommand{\newreptheorem}[2]{
\newenvironment{rep#1}[1]{
 \def\rep@title{#2 \ref{##1}}
 \begin{rep@theorem}\itshape}
 {\end{rep@theorem}}}
\makeatother
\theoremstyle{plain}

\newreptheorem{theorem}{Theorem}

\makeatletter
\newtheorem*{rep@claim}{\rep@title}
\newcommand{\newrepclaim}[2]{
\newenvironment{rep#1}[1]{
 \def\rep@title{#2 \ref{##1}}
 \begin{rep@claim}\itshape}
 {\end{rep@claim}}}
\makeatother
\theoremstyle{plain}

\newrepclaim{claim}{Claim}

\begin{document}

\title{Query strategies for priced information, revisited \vspace{10pt} }

\author{Guy Blanc \vspace{10pt} \\ \hspace{-5pt}{\sl Stanford} \and \hspace{10pt} Jane Lange \vspace{10pt} \\
\hspace{8pt} {\sl MIT} \vspace{15pt} 
\and Li-Yang Tan \vspace{10pt} \\ 
\hspace{-5pt}{\sl Stanford} }  

 \date{\small{\today}}

\maketitle

\begin{abstract}
We consider the problem of designing {\sl{query strategies for priced information}}, introduced by Charikar et al.   In this problem the algorithm designer is given a function~$f : \zo^n \to \bits$  and a price associated with each of the $n$ coordinates.  The goal is to design a query strategy for determining $f$'s value on unknown inputs for minimum cost. 

Prior works on this problem have focused on specific classes of functions.  We analyze a simple and natural strategy that applies to all functions $f$, and show that its performance relative to the optimal strategy can be expressed in terms of a basic complexity measure of $f$, its {\sl influence}.  For $\eps \in (0,\frac1{2})$, writing $\mathsf{opt}$ to denote the expected cost of the optimal strategy that errs on at most an $\eps$-fraction of inputs, our strategy has expected cost $\mathsf{opt} \cdot \Inf(f)/\purple{\eps^2}$ and also errs on at most an $O(\eps)$-fraction of inputs. This connection yields new guarantees that complement existing ones for a number of function classes that have been studied in this context, as well as new guarantees for new classes.

Finally, we show that improving on the parameters that we achieve will require  making progress on the longstanding open problem of properly learning decision trees. 
\end{abstract}

 \thispagestyle{empty}
\newpage 

\setcounter{page}{1}

\section{Introduction}

We consider the problem of designing {\sl{query strategies for priced information}}, introduced by Charikar, Fagin, Guruswami, Kleinberg, Raghavan, and Sahai~\cite{CFGKRS00}.  In this problem, the algorithm designer is given a function $f : \zo^n \to \{\pm 1\}$ and a {\sl price} associated with each of the $n$ coordinates.   The goal is to design a {\sl query strategy} for evaluating $f$ on an unknown input $\underline{x} \in \zo^n$ for minimum cost.  A query strategy $\mathcal{S}$ is an adaptive sequence of queries to the coordinates of $\underline{x}$, where a cost of~$c_i$ is incurred to reveal the $i$-th coordinate of $\underline{x}$, at the end of which $\mathcal{S}$ is expected to output $f(\underline{x})$.   

The introduction of prices to the query model adds a rich new dimension to one of the simplest models of computation.   Quoting Charikar et al., ``We find that the inclusion of arbitrary prices on [input coordinates] gives the problem a much more complex character, and leads to query algorithms that are novel and non-obvious"~\cite[p.~583]{CFGKRS00}.  Indeed, following the work of~\cite{CFGKRS00}, there have been a large number of works studying various formulations of this problem, and within these formulations, designing strategies for a wide variety of function classes~\cite{GK01,KKM05,CL05,CL05-ESA,CL06,CL08,CGLM11,CM11,CL11,DHK14,CLS17,AHKU17,GGHK18,BDHK18,Hel18}.

The formulation that we work with is a slight variant of the one introduced by Deshpande, Hellerstein, and Kletenik~\cite{DHK14}.   We are given as input a succinct representation of a function $f : \zo^n \to \{\pm 1\}$, a cost vector $c \in \N^n$ where $c_i$ is the cost associated with the $i$-th input coordinate, and an error parameter $\eps \in (0,\frac1{2})$.  For a strategy $\mathcal{S}$ and an input $\underline{x} \in \zo^n$ to $f$, we write $\cost_c(\mathcal{S},\underline{x})$ to denote the total cost incurred by $\mathcal{S}$ on input~$\underline{x}$, and $\Delta_c(\mathcal{S}) \coloneqq \Ex[\cost_c(\mathcal{S},\bx)]$ to denote the expected cost of $\mathcal{S}$ on a uniform random input $\bx \sim \zo^n$.   We write $\mathcal{S}(\underline{x})$ to denote the output of $\mathcal{S}$ on $\underline{x}$.  

Our goal is to design {\sl accurate}, {\sl computationally efficient}, and {\sl cost efficient} strategies:

\begin{itemize}[leftmargin=0.5cm]
\item[$\circ$] {\sl Accuracy}: We say that $\mathcal{S}$ is an {\sl $\eps$-error strategy for $f$} if $\Prx[\mathcal{S}(\bx) \ne f(\bx)] \le \eps,$ where $\bx\sim\zo^n$ is uniform random. 
\item[$\circ$] {\sl Computational efficiency}: We say that $\mathcal{S}$ is {\sl computationally efficient} if every next query, and the value of $\mathcal{S}(\underline{x})$ after the final query, can be computed in time that is polynomial in the input length (i.e.~the representation sizes of $f$, the cost vector $c$, and $\eps$).  
\item[$\circ$] {\sl Cost efficiency:}  Writing $\opt = \opt(f,c,\eps)$ to denote the minimum expected cost among all $\eps$-error strategies for $f$, our goal is to design a strategy $\mathcal{S}$ with expected cost as close to $\opt$ as possible. (While we require our strategy $\mathcal{S}$ to be computationally efficient, we do not require the minimum-cost strategy that witnesses $\opt$ to be computationally efficient.) 
\end{itemize} 

In the case of randomized strategies $\mathcal{S}$, both expected cost and accuracy are also measured with respect to the randomness of $\mathcal{S}$ (in addition to the randomness of the input $\bx$). 

\pparagraph{\violet{This work: A strategy for all functions.}}   Prior works on various formulations of the problem of querying priced information  focused on different specific classes of functions: conjunctions and disjunctions~\cite{KKM05}, various subclasses of DNF formulas~\cite{KKM05,DHK14,AHKU17}, halfspaces and their generalizations~\cite{FP04,DHK14,GGHK18}, {\sc And}-{\sc Or} trees~\cite{CFGKRS00}, symmetric functions~\cite{GGHK18,Hel18}, and so on.  These works designed and analyzed strategies that are tailored to each of these classes.   

The main contribution of our work is to analyze a simple and natural strategy that applies to  all functions $f$, and to show that its expected cost relative to $\opt$ can be bounded in terms of a basic and well-studied complexity measure of boolean functions, $f$'s {\sl total influence}, denoted $\Inf(f)$, which we will soon define in the next subsection:

\begin{theorem}[Our main result]
\label{thm:main} 
Let $f : \zo^n \to \{\pm 1 \}$ be a function, $c \in \N^n$ be a cost vector, and $\eps \in (0,\frac1{2})$ be an error parameter.  There is an efficient $O(\eps)$-error randomized strategy $\mathcal{S}^\star$ for $f$ with expected cost $\Delta_c(\mathcal{S}^\star) \le \opt\cdot \Inf(f)/\purple{\eps^2}$.
\end{theorem} 

\violet{\Cref{thm:main} enables us to leverage a large body of results on the total influence of various function classes to derive guarantees on the performance of our strategy when invoked on functions from the corresponding class.  As we will show in~\Cref{sec:example}, this yields new guarantees that add to existing ones for a number of function classes that have been studied in this context, as well as new guarantees for new classes.}

\violet{We complement~\Cref{thm:main} with two lower bounds, one showing the optimality of the analysis of our strategy $\mathcal{S}^\star$, and the other showing that designing a different strategy that improves upon the parameters that we achieve with $\mathcal{S}^\star$ will require making progress on the longstanding open problem of properly learning decision trees.}



\subsection{Background and motivation}
\label{sec:background} 
To motivate our strategy, we first consider the unit-cost setting ($c_i = 1$ for all coordinates $i \in [n]$) and ask the simplest possible question: ``{\sl What should our first query be?}\,"   An answer to this question naturally suggests a corresponding greedy strategy: restrict $f$ according to the response to the first query, and recurse, i.e.~ask the same question of the restricted function.  Intuitively, we would like to query the most ``important" coordinate of the unknown input $\underline{x}$, the one for which knowing its value goes the furthest towards determining $f$'s value on $\underline{x}$, and hence minimizes the number of subsequent queries.  

Determining the {\sl optimal} such query is $\mathsf{NP}$-hard---minimization of decision trees is $\mathsf{NP}$-hard~\cite{LR76,ZB00}, even to approximate~\cite{Sie08,AH12}---so we seek {\sl efficiently computable} ``proxies" for how important  a coordinate is.  For example, if $f$ is itself represented as a decision tree, a natural candidate would be the coordinate that is queried in the most branches of the tree;  if $f$ is represented as a DNF formula, a natural candidate would be the coordinate that appears in the most terms; if $f$ is represented as a halfspace $\sign(\sum_{i=1}^n w_i x_i - \theta)$, a natural candidate would be the coordinate $i$ for which the corresponding weight $w_i$ has the largest magnitude; and so on.   Indeed, a number of previous works analyze these very strategies, as well as generalizations of them to the setting of general costs.  Note that these strategies are all {\sl representation specific}: for example, the strategy for DNF formulas cannot be carried out if $f$ is represented as a halfspace.

\pparagraph{Influence and a representation-independent strategy.}  In contrast to the aforementioned examples, our strategy will be based on a notion of ``importance"  that is {\sl representation independent}, and hence applicable to all functions $f$: 

\begin{definition}[Influence]
\label{def:influence} 
The {\sl influence of the variable $x_i$ on a function $f : \zo^n \to \{\pm 1\}$} is defined to be \[ \Inf_i(f) \coloneqq \Prx_{\bx\sim \zo^n}[f(\bx) \ne f(\bx^{\oplus i})],\] 
where $\bx$ is drawn uniformly at random, and $\bx^{\oplus i}$ denotes $\bx$ with its $i$-th coordinate flipped.   The {\sl total influence} of $f$ (or simply its {\sl influence}) is defined to be $\Inf(f) \coloneqq \sum_{i=1}^n \Inf_i(f)$.  \end{definition}


\violet{The influence of variables and the total influence of functions are by now mainstay notions in theoretical computer science, arising in a range of areas across both algorithms and complexity theory; for a detailed treatment, see O'Donnell's book~\cite{ODBook}.} With this definition in hand, we can now state a simple strategy, $\mathcal{S}^{\star}_{\textsc{UnitCost}}$, for the unit-cost setting:  query the coordinate of $\underline{x}$ that corresponds to the most influential variable of $f$, restrict $f$ accordingly, and recurse.  
\violet{Our actual strategy for the setting of general costs, which we will soon describe in the next subsection, is a natural generalization of $\mathcal{S}^{\star}_{\textsc{UnitCost}}$ that takes into account the both the influence of a variable and its cost.} 


Despite the apparent natural fit, the notion of influence, along with its associated suite of powerful techniques and results, seem to have been surprisingly underutilized for the problem of querying priced information. A main goal of this work is to bridge this connection.

\subsection{Our results} 
\label{sec:our-results}

{\bf Our strategy.} We analyze a natural generalization of $\mathcal{S}^\star_{\textsc{UnitCost}}$ to the setting of general costs: rather than iteratively querying the most influential variable, our general strategy $\mathcal{S}^\star$ will query the variable with the highest {\sl cost-weighted} influence: 

\begin{figure}[H]
    \captionsetup{width=.9\linewidth}
    \begin{tcolorbox}[colback = white,arc=1mm, boxrule=0.25mm]
        \vspace{3pt} 
        
        Query strategy $\mathcal{S}^\star(f,c,\eps)$ for an unknown input $\underline{x}$:  \vspace{6pt} 
        
         \ \ Define $B \coloneqq \opt \cdot \Inf(f)/\purple{\eps^2}$ \vspace{4pt} 

        \ \ Initialize {\tt counter} $\coloneqq 0$\vspace{4pt}  
        
        \ \ while {\tt counter} $ < B$:  \vspace{4pt} 
        
        
        \ \ \ \ \ \ Let $i \in [n]$ be the coordinate with the largest cost-weighted influence:
        \[ \frac{\Inf_i(f)}{c_i} \ge \frac{\Inf_j(f)}{c_j} \quad \text{for all $j\in [n]$}.\] 
        \ \ \ \ \ \ if ${\tt counter} + c_i \geq B$:\vspace{4pt} 
        
        \ \ \ \ \ \ \ \ \ \ \ \ With probability $\frac{B - {\tt counter}}{c_i}$:\vspace{4pt}
        
        \ \ \ \ \ \ \ \ \ \ \ \ \ \ \ \ \ \ Query the $i$-th coordinate of $\underline{x}$ and output $\sign(\E[f_{x_i = \underline{x}_i}])$. \vspace{4pt}

        \ \ \ \ \ \ \ \ \ \ \ \ With probability $1 - \frac{B - {\tt counter}}{c_i}$:\vspace{4pt}
        
        \ \ \ \ \ \ \ \ \ \ \ \ \ \ \ \ \ \ Output $\sign(\E[f])$. \vspace{4pt}
        
        \ \ \ \ \ \ Query the $i$-th coordinate of $\underline{x}$, and restrict $f$ accordingly: $f = f_{x_i = \underline{x}_i}$  \vspace{4pt}
        
        \ \ \ \ \ \ {\tt counter+=} $c_i$  
        
        \end{tcolorbox}
    \caption{Our query strategy $\mathcal{S}^\star$ takes as input a representation of $f : \zo^n\to \{\pm 1\}$, a cost vector $c\in \N^n$, and an error parameter $\eps$.  It has query access to an unknown input~$\underline{x}$, and its goal is to compute $f(\underline{x})$ in a computationally- and cost-efficient manner.}
    \label{fig:our-strategy}
\end{figure}


 






As described in~\Cref{fig:our-strategy}, it is not yet clear that $\mathcal{S}^\star$ can be implemented in a computationally-efficient manner, for two reasons. First, in order to define the parameter `$B$', the strategy has to ``know" the value of $\opt$ and $\Inf(f)$; second, in order to determine which variable has the highest cost-weighted influence, the strategy has to be able to compute variable influences efficiently.  The latter is handled easily using the fact that the influence of variables can be estimated to high accuracy via a standard random sampling argument (recall~\Cref{def:influence}).  The former can also be handled in a fairly straightforward manner:  we describe an efficient randomized algorithm for computing `$B$' in~\Cref{sec:stop}.

\pparagraph{Our guarantees on the performance of $\mathcal{S}^\star$.} We now introduce terminology to refer to the performance of the {\sl optimal} strategy, which will serve as our benchmark for the evaluation of our strategy $\mathcal{S}^\star$:  

\begin{definition}[$\opt$]
\label{def:opt} 
Let $f : \zo^n \to \{ \pm 1\}$ be a function, $c \in \N^n$ be a cost vector, and $\eps \in (0,\frac1{2})$ be an error parameter.  We write $\opt(f,c,\eps)$ to denote the expected cost, relative to $c$, of the optimal $\eps$-error strategy for $f$: 
\[ \opt(f,c,\eps) \coloneqq \min\{\, \Delta_c(\mathcal{S}) \colon \text{$\mathcal{S}$ is an $\eps$-error strategy for $f$}\,\},\] 
where $\Delta_c(\mathcal{S}) \coloneqq \E[\cost_c(\mathcal{S},\bx)]$.  When $f,c,$ and $\eps$ are clear from context, we simply write $\opt$. 
\end{definition} 

We can now (re)state our main result: 

\begin{reptheorem}{thm:main}
Let $f : \zo^n \to \{\pm 1 \}$ be a function, $c \in \N^n$ be a cost vector, and $\eps \in (0,\frac1{2})$ be an error parameter.  The strategy $\mathcal{S}^\star$ is an efficient randomized $O(\eps)$-error strategy for~$f$ with expected cost $\Delta_c(\mathcal{S}^\star) \le \opt\cdot \Inf(f)/\purple{\eps^2}$.   
\end{reptheorem}


An appealing feature of our strategy $\mathcal{S}^\star$ is that it in fact only requires blackbox query access to~$f$, rather than access to a full representation of $f$ as in previous works.  This is a corollary of the already-discussed fact that influence is a representation-independent notion. 

\pparagraph{The optimality of our analysis, and the possibility of improved strategies.} We complement~\Cref{thm:main} with two lower bounds.  The first concerns the optimality of our analysis of the strategy $\mathcal{S}^\star$:

\begin{claim}[Optimality of our analysis] 
\label{claim:optimality}
Consider the unit-cost setting, $c = 1^n$.  For any $\eps \in (0, \frac{1}{4})$, there is a function $f : \zo^n \to \{ \pm 1\}$ such that our strategy $\mathcal{S}^\star$ has to incur expected cost $\opt \cdot \tilde{\Omega}_\eps(\Inf(f))$ in order to achieve error $\eps$. 
\end{claim} 

While~\Cref{claim:optimality} shows that the parameters of~\Cref{thm:main} are optimal for our strategy $\mathcal{S}^\star$, it does not rule out the possibility that a {\sl different} strategy achieves better parameters than those of~\Cref{thm:main}.  Our next claim shows that if such a strategy, like $\mathcal{S}^\star$, only uses blackbox query access to the function, then it would also improve upon the state of the art for {\sl properly learning decision trees}, a central open problem in learning theory:\footnote{A learning algorithm is {\sl proper} if it returns a hypothesis that falls within the concept class.  An algorithm for properly learning decision trees is therefore one that returns a decision tree as its hypothesis.}  

\begin{claim}[Improving on~\Cref{thm:main} $\Rightarrow$ New algorithms for properly learning decision trees]  
\label{claim:breakthrough}
Consider the unit-cost setting, $c = 1^n$.  Suppose that for all functions $f : \zo^n \to \{\pm 1\}$, there is an efficient $O(\eps)$-error strategy $\mathcal{S}$ for $f$ with expected cost $\opt \cdot o_\eps(\Inf(f))$.   Furthermore, suppose that $\mathcal{S}$ only requires blackbox query access to $f$.  Then there is an $n^{o(\log n)}$-time query algorithm for properly PAC learning polynomial-size decision trees under the uniform distribution.  
\end{claim} 

The current fastest algorithm for properly learning decision trees runs in time $n^{O(\log n)}$ and is due to Ehrenfeucht and Haussler~\cite{EH89}. The algorithm of~\Cref{claim:breakthrough} would therefore represent the first improvement over the state of the art in three decades.\footnote{As our proof of~\Cref{claim:breakthrough} will show, the claimed implication holds even under fairly strong  assumptions about $f$: it suffices to design such a strategy for $f$'s that are computable by $\poly(n)$-size decision trees.}  

Even if one considers an $n^{o(\log n)}$ time algorithm for properly learning decision trees as being out of reach for now (or simply impossible),~\Cref{claim:breakthrough} points to a concrete avenue for future work: to improve the parameters of~\Cref{thm:main} by exploiting the fact that the strategy has access to a representation of $f$.  For many representation classes, it is reasonable to expect that the representation of $f$ simplifies significantly as queries are made and $f$ is restricted accordingly.  This is an aspect that is not taken into account in our current analysis, and could lead to improved parameters; we view this as a concrete avenue for future work.


\subsection{Example applications of~\Cref{thm:main} \violet{and comparison with prior work}} 
\label{sec:example} 

In this subsection we give a few example applications of~\Cref{thm:main}. By combining it with known bounds on the total influence of various function classes, we show that our strategy $\mathcal{S}^\star$ \violet{achieves new guarantees that complement existing ones for a number of function classes that have been studied in this context, as well as new guarantees for new classes.}

We begin with {\sl DNF formulas}, a well-studied function class in this literature on query strategies for priced information.  Combining~\Cref{thm:main} with a classic result of Boppana~\cite{Bop97} on the total influence of DNF formulas, we get: 

\begin{corollary}[DNF formulas] 
\label{cor:DNF} 
Let $f: \zo^n \to \{ \pm 1\}$ be a $t$-term DNF formula.  For any cost vector $c \in \N^n$ and error parameter $\eps \in (0,\frac1{2})$, our strategy $\mathcal{S}^\star$ is an $O(\eps)$-error strategy for $f$ with expected cost $\Delta_c(\mathcal{S}^\star) \le \opt \cdot O(\log t)/\purple{\eps^2}$. 
\end{corollary} 

A number of prior works proposed and analyzed query strategies for various subclasses of DNF formulas.   Kaplan, Kushilevitz, and Mansour~\cite{KKM05} studied read-once DNF formulas and monotone ``CDNF formulas" (small DNF formulas that also have small CNF representations); Deshpande, Hellerstein, and Kletenik~\cite{DHK14} studied general (not-necessarily-monotone) CDNF formulas; Allen, Hellerstein, Kletenik, and {\"U}nl{\"u}yurt~\cite{AHKU17} studied monotone $k$-term DNF formulas and monotone width-$k$ DNF formulas.   To our knowledge,~\Cref{cor:DNF} is the first result that applies to all DNF formulas.  (While~\cite{KKM05,DHK14} also achieve an approximation ratio of $O(\log t)$ for CDNF formulas, \Cref{cor:DNF} is incomparable to these results since~\cite{KKM05,DHK14} work with different formulations of the problem that are incomparable to ours; we give an overview of these alternative formulations in the next subsection.)

A natural generalization of DNF formulas is the class of {\sl small-depth circuits} over the standard $\{ {\textsc{And}}, {\text{Or}}, {\textsc{Not}} \}$ basis, with DNF formulas being the special case of depth-two circuits.  The total influence of such circuits is also well understood~\cite{LMN93,Bop97}; combining~\Cref{thm:main} with these bounds yields the first non-trivial guarantee for functions computable by small-depth circuits:

\begin{corollary}[Small-depth circuits; generalization of~\Cref{cor:DNF}] 
Let $f: \zo^n \to \{ \pm 1\}$ be a size-$t$ depth-$d$ circuit.  For any cost vector $c \in \N^n$ and error parameter $\eps \in (0,\frac1{2})$, our strategy $\mathcal{S}^\star$ is an $O(\eps)$-error strategy for $f$ with expected cost $\Delta_c(\mathcal{S}^\star) \le \opt \cdot O(\log t)^{d-1}/\purple{\eps^2}$. 
\end{corollary}

  Note that the approximation factor of $O(\log t)^{d-1}$ is $\polylog(n)$ for $t = \poly(n)$ and $d = O(1)$, the setting of parameters that corresponds to the circuit class $\mathsf{AC}^0$.  
 
 Like DNF formulas, the class of {\sl halfspaces} $f(x) = \sign(w\cdot x-\theta)$ has also received significant attention in the literature~\cite{FP04,DHK14,GGHK18}.  However, to our knowledge, there has been no non-trivial strategies for the {\sl intersections of halfspaces}: functions of the form $f(x) = h_1(x) \wedge \cdots \wedge h_m(x)$ where each $h_i(x)$ is a halfspace.  These functions correspond to $\zo$-integer programs with $m$ constraints and also to $m$-facet polytopes, and are a basic object of study in computer science. Combining~\Cref{thm:main} with a result of Kane~\cite{Kan14-polytopes} on the total influence of such functions, we get:

\begin{corollary}[Intersections of halfspaces] 
\label{cor:polytopes} 
Let $f: \zo^n \to \{ \pm 1\}$ be the intersection of $m$ halfspaces.  For any cost vector $c \in \N^n$ and error parameter $\eps \in (0,\frac1{2})$, our strategy $\mathcal{S}^\star$ is an $O(\eps)$-error strategy for $f$ with expected cost $\Delta_c(\mathcal{S}^\star) \le \opt \cdot O(\sqrt{n\log m})/\purple{\eps^2}$. 
\end{corollary}
  
 Note that because of the logarithmic dependence on `$m$' in~\Cref{cor:polytopes}, the approximation factor remains $\tilde{O}(\sqrt{n})$ for any $m = \poly(n)$. 
 
  As our final example, we combine~\Cref{thm:main} with bounds on the total influence of {\sl monotone} functions~\cite{BT96} to get: 

\begin{corollary}[Monotone functions]
Let $f: \zo^n \to \{ \pm 1\}$ a monotone function.  For any cost vector $c \in \N^n$ and error parameter $\eps \in (0,\frac1{2})$, our strategy $\mathcal{S}^\star$ is an $O(\eps)$-error strategy for $f$ with expected cost $\Delta_c(\mathcal{S}^\star) \le \opt \cdot O(\sqrt{n})/\purple{\eps^2}$. 
\end{corollary} 

Monotonicity is a natural property of functions, and monotone functions are intensively studied in both algorithms and complexity theory.  The number of distinct $n$-variate monotone functions is {\sl doubly} exponential in $n$, making it a much larger class than the other examples discussed above (which have sizes that are only singly exponential in $n$).

\subsection{\violet{Other models for querying priced information}} 
\label{sec:prior} 

In this subsection we give an overview of the other formulations of the general problem of querying priced information that have been studied in the literature.

Charikar, Fagin, Guruswami, Kleinberg, Raghavan, and Sahai~\cite{CFGKRS00} employed the framework of {\sl competitive analysis}: instead of comparing the expected cost of a strategy $\mathcal{S}$ to that of the optimal strategy (as we do),~\cite{CFGKRS00} compares, for each possible input $\underline{x}$, the cost incurred by $\mathcal{S}$ on $\underline{x}$ to the cost of the cheapest {\sl proof} that certifies $f$'s value on $\underline{x}$. Such a proof is a subset of coordinates $S \sse [n]$ for which restricting $f$ by $\underline{x}_S$ fixes its value to $f(\underline{x})$.  Their goal is to design strategies that minimize this ratio for a worst-case $\underline{x}$.  Within this formulation,~\cite{CFGKRS00} designed strategies for functions $f$ computable by ``{\sc And}-{\sc Or} trees": formulas with internal nodes labeled by binary $\{ {\textsc{And}}, {\textsc{Or}}\}$ gates, and leaves labeled by distinct variables.   There have been quite a number of papers that follow up on~\cite{CFGKRS00}, working within their framework of competitive analysis~\cite{GK01,CL05,CL05-ESA,CL06,CL08,CGLM11,CM11,CL11,CLS17}. 

The formulation in the work of Kapalan, Kushilevitz, and Mansour~\cite{KKM05} can be viewed as an average-case version of~\cite{CFGKRS00}'s: they compare the expected cost of a strategy to the expected cost, with respect to a random input~$\bx$, of the cheapest proof that certifies $f$'s value on~$\bx$.  While we only require the strategy to compute the value of $f(\underline{x})$,~\cite{KKM05} requires it to additionally output an ``explanation" of $f(\underline{x})$: in the case of CDNF formulas that they study, this would be a term in the DNF representation that $\underline{x}$ satisfies when $f(\underline{x}) = 1$, and a clause in the CNF representation that $\underline{x}$ falsifies when $f(\underline{x}) = -1$.  Another difference is that~\cite{KKM05} also considers the possibility of non-uniform distributions over inputs (either product distributions or an arbitrary distribution), whereas we focus  on the uniform distribution.   In addition to monotone CDNF formulas,~\cite{KKM05} also gives results for disjunctions and read-once DNF formulas.

The formulation that we work with is most similar to the one introduced by Deshpande, Hellerstein, and Kletenik~\cite{DHK14} and further studied in~\cite{AHKU17,GGHK18,Hel18,BDHK18}.  A key difference is that these works (and also~\cite{KKM05} discussed above) focus on {\sl zero-error} strategies.  We on the other hand allow strategies to err on a small fraction of inputs, and we benchmark them against the optimal strategy that is also allowed to error on a small fraction of inputs; this makes our formulation incomparable to theirs. Another difference is that these works additionally consider the setting of general product distributions over inputs, whereas we focus  on the uniform distribution.
The function classes studied in these papers include CDNF formulas, monotone $k$-term DNF formulas and width-$k$ DNF formulas, halfspaces, and symmetric functions.

\subsection{Relation to~\cite{BLT20}} 
\violet{The starting point of our work is the paper~\cite{BLT20}, which analyzes $\mathcal{S}^\star_{\textsc{UnitCost}}$ from the vantage point of learning theory.  Implicit in the analysis of~\cite{BLT20} is a bound on expected cost of $\mathcal{S}^\star_{\textsc{UnitCost}}$ relative to $\opt(f,c = 1^n, \eps=0)$---recalling~\Cref{def:opt}, this is the special case of unit costs, and the laxer setting where the performance of $\mathcal{S}^\star_{\textsc{UnitCost}}$ is compared to the optimal {\sl zero-error} strategy.}  

\violet{Our work builds on and extends the techniques in~\cite{BLT20}; three key  differences between the setup in~\cite{BLT20} and ours are:} 

\begin{enumerate}[leftmargin=0.6cm]
\item {\sl The incorporation of prices}: in~\cite{BLT20}, the goal is to construct an approximate decision tree representation for an unknown function $f$ with as small an average depth as possible.\footnote{\cite{BLT20} actually state their results in terms of decision tree {\sl size}, a weaker complexity measure than average depth.  However, their analysis extends to handle average depth as the complexity measure.} This corresponds to the unit-cost setting, since every query contributes equally to the average depth of a tree, regardless of {\sl which} coordinate is queried.  A main challenge we have to overcome in this work is the one introduced by the overall setting of querying {\sl priced} information: our goal is to  minimize the {\sl cost-weighted} average depth of a tree, where certain queries contribute more than others. 
\item {\sl Benchmarking against the optimal exact versus approximate decision tree/strategy}: in~\cite{BLT20}, the average depth of the {\sl approximate} decision tree representation built by $\mathcal{S}^\star_{\textsc{UnitCost}}$ is compared to that of the optimal {\sl exact} decision tree representation of $f$. The followup work~\cite{BLT3} shows that if $f$ is assumed to be {\sl monotone}, then the average depth of the approximate decision tree representation built by $\mathcal{S}^\star_{\textsc{UnitCost}}$ can in fact be bounded by the optimal approximate decision tree representation of $f$.  In this work, we bound the expected cost of our strategy $\mathcal{S}^\star$, which we prove achieves error $O(\eps)$, in terms of that the optimal $\eps$-error strategy for $f$, without relying on any assumptions about $f$.  
\item {\sl Local versus global algorithm}: 
\violet{in the learning-theoretic setting of~\cite{BLT20}, the  learning algorithm has a ``global" view of the entire decision tree representation of $f$ as it constructs the tree.  In contrast, in the context of querying priced information, the goal of the query strategy is to compute $f(\underline{x})$ for a {\sl specific} unknown input $\underline{x}$ (constructing the entire decision tree would render the strategy inefficient).  This constrains the strategy to proceed in a ``local" fashion---this corresponds to only constructing the single root-to-leaf path of the overall decision tree, the unique one that is consistent with $\underline{x}$---which introduces technical challenges. For example, in~\cite{BLT20} the termination condition of their learning algorithm can---and does---depend on the global structure of the decision tree, whereas the termination condition of our strategy cannot. }

\end{enumerate} 

We are hopeful that the techniques we have introduced to handle these differences will find further applications back in learning theory, the domain of~\cite{BLT20}; the extension described in Item~1 above, in particular, is related to the setting of learning with {\sl attribute costs}~\cite{KKM05}.  We view this as an interesting avenue for future work.  


\section{Preliminaries} 

We use {\bf boldface} (e.g.~$\bx \sim \zo^n$) to denote random variables.  Unless otherwise stated, all probabilities and expectations are with respect to the uniform distribution.  Given a function $f : \zo^n \to \bits$, we write $\bias(f)$ to denote $\min_{b\in \bits}\{ \Pr[f(\bx)=b]\}$, noting that $\bias(f) = \Pr[f(\bx)\ne \sign(\E[f])]$, the error incurred by approximating $f$ with a constant function.  Given a function $f : \zo^n \to \bits$ and a strategy $\mathcal{S}$, we write $\error_f(\mathcal{S}) \coloneqq \Pr[\mathcal{S}(\bx) \ne f(\bx)]$ to denote the error of $\mathcal{S}$ as a strategy for $f$. Given a strategy $\mathcal{S}$ and cost vector $c \in \N^n$, we write $\Delta_c(\mathcal{S})$ as shorthand for $\E[\cost_c(\mathcal{S},\bx)]$.

\purple{
It will be useful for us to define a {\sl worst-case} analogue of $\opt$.
\begin{definition}[$\overline{\opt}$]
\label{def:opt} 
Let $f : \zo^n \to \{ \pm 1\}$ be a function, $c \in \N^n$ be a cost vector, and $\eps \in (0,\frac1{2})$ be an error parameter.  We write $\overline{\opt}(f,c,\eps)$ to denote the worst-case cost, relative to $c$, of the optimal $\eps$-error strategy for $f$: 
\[ \opt(f,c,\eps) \coloneqq \min\{\, \overline{\Delta}_c(\mathcal{S}) \colon \text{$\mathcal{S}$ is an $\eps$-error strategy for $f$}\,\},\] 
where $\overline{\Delta}_c(\mathcal{S}) \coloneqq \max_{x \in \zo^n} \cost_c(\mathcal{S},x)$.  When $f,c,$ and $\eps$ are clear from context, we simply write $\overline{\opt}$. 
\end{definition} }

\subsection{Properties of strategies}
Recall that a strategy, often denoted $\mathcal{S}$, makes a series of adaptive queries to the coordinates of its input and then returns a value as a function of those query outputs. A \textit{deterministic} strategy always makes the same queries and outputs the same value on any fixed input while \textit{randomized} strategies can use randomness as well as previous query outputs to make those decisions.

We call a history of queries and their outputs a \textit{restriction}, often denoted $\alpha$. We say that some $x \in \zo^n$ is \textit{consistent} with $\alpha$ if coordinate queries to $x$ would give the results stored in $\alpha$. For any function $f: \zo^n \to \bits$ and restriction $\alpha$, we use $f_{\alpha}$ to refer to $f$ with it's domain restricted to just those inputs consistent with $\alpha$. For example,
\begin{align*}
    \Ex[f_{\alpha}] \coloneqq \Ex_{\bx \text{ consistent with $\alpha$}}[f(\bx)].
\end{align*}

On any input $x \in \zo^n$ and strategy $\mathcal{S}$, we use $\leaf(\mathcal{S}, x)$ to refer to the last restriction $\mathcal{S}$ makes when given $x$ as input. The value $\mathcal{S}$ outputs is a (potentially randomized) function of $\leaf(\mathcal{S}, x)$. Similarly, we use $\paths(\mathcal{S}, x)$ to denote the sequence of restrictions $\mathcal{S}$ creates given input $x$ \textit{before} it reaches $\leaf(\mathcal{S}, x)$. Given any restriction $\alpha$, we will use $\query(\mathcal{S}, \alpha)$ to denote the coordinate that $\mathcal{S}$ queries given a query history of $\alpha$. In randomized strategies, $\query(\mathcal{S}, \alpha)$ and $\paths(\mathcal{S}, x)$ can be random variables. In this notation, average cost is computed as
\begin{align*}
    \Delta_c(\mathcal{S}) = \Ex_{\bx \in \zo^n}\left[\sum_{\alpha \in \paths(\mathcal{S}, x)} c_{\query(\mathcal{S},\alpha)} \right].
\end{align*}
For any function $f: \zo^n \to \bits$, a strategy $\mathcal{S}$ is said to be \textit{$f$-consistent} if it has the minimum error with respect to $f$ of all strategies making the same queries. Equivalently, $\mathcal{S}$ is $f$-consistent if, for all $x \in \zo^n$
\begin{align*}
    \mathcal{S}(x) = \sign\left(\Ex[f_{\leaf(\mathcal{S}, x)}] \right).
\end{align*}
Any $f$-consistent strategy $\mathcal{S}$ has error
\begin{align*}
     \error_f(\mathcal{S}) = \Ex_{\bx \in \zo^n}\left[\bias(f_{\leaf(\mathcal{S}, x)}) \right].
\end{align*}
The \textit{average influence} of a strategy with respect to a function, denote $\AvgInf_f(\mathcal{S})$, is
\begin{align*}
    \AvgInf_f(\mathcal{S})\coloneqq \Ex_{\bx \in \zo^n}\left[\Inf(f_{\leaf(\mathcal{S}, x)}) \right].
\end{align*}
Because, for any function $f$, $\bias(f) \leq  \Inf(f)$, if $\mathcal{S}$ is $f$-consistent, $\AvgInf_f(\mathcal{S})$ upper bounds $\error_f(\mathcal{S})$. That fact, combined with the following properties showing how average influence decreases, make it a useful progress measure.

\begin{fact}[Influences split naturally]
    \label{fact:influence split}
    For any function $f: \zo^n \to \bits$ and $i \in [n]$,
    \begin{align*}
        \Inf(f) = \Inf_i(f) + \frac{1}{2}(\Inf(f_{x_i = 0}) + \Inf(f_{x_i = 1})).
    \end{align*}
\end{fact}

A simple corollary of \Cref{fact:influence split} is that it is possible to determine the average influence of a strategy by looking at the queries it makes.
\begin{corollary}[Average influence of strategy]
\label{cor:average influence}
    For any function $f: \zo^n \to \bits$ and strategy $\mathcal{S}$,
    \begin{align*}
        \AvgInf_f(\mathcal{S}) = \Inf(f) - \Ex_{\bx \in \zo^n}\left[\sum_{\alpha \in \paths(\mathcal{S}, \bx)} \Inf_{\query(S, \alpha)}(f_\alpha) \right].
    \end{align*}
\end{corollary}
For deterministic strategies, \Cref{cor:average influence} follows from \Cref{fact:influence split} via induction on the maximum number of queries the strategy makes. For randomized strategies, it can be proved by observing that any randomized strategy is a mixture of deterministic strategies and then using linearity of expectation.

\section{Warmup} 

Before proving~\Cref{thm:main}, we begin with a warmup, the proof of which is significantly simpler and yet already illustrates a few of the key new ideas. Rather than bounding the expected cost of our strategy in terms of $\opt \coloneqq \opt(f,c,\eps)$, in this warmup we will bound it in terms of the laxer quantity $\purple{\overline{\opt'}} \coloneqq \purple{\overline{\opt}}(f,c,0)$.  That is, $\purple{\overline{\opt'}}$ is the \purple{worst-case} cost of the optimal {\sl zero-error} strategy $\mathcal{S}_\opt$ for $f$ (i.e.~$\mathcal{S}_\opt(x) = f(x)$ for all $x\in \zo^n$).   

\begin{theorem}[Warmup for~\Cref{thm:main}] 
\label{thm:warmup} 
Let $f : \zo^n \to \{\pm 1\}$ be a function, $c \in \N^n$ be a cost vector, and $\eps \in (0,\frac1{2})$.  There is an efficient $O(\eps)$-error randomized strategy $\mathcal{S}^\diamond$ for $f$ with expected cost $\Delta_c(\mathcal{S}^\diamond) \le \purple{\overline{\opt'}}\cdot \Inf(f)/\eps$. 
\end{theorem} 

\subsection{Proof of~\Cref{thm:warmup}}

For a strategy $\mathcal{S}$ and coordinate $i\in [n]$, we write $\delta_i(\mathcal{S})$ to denote the probability that $\mathcal{S}$ queries $\bx_i$ where $\bx \sim \zo^n$ is uniform random.  We have that 
\begin{equation}  \sum_{i=1}^n \delta_i(\mathcal{S})\cdot c_i = \E[\cost_c(\mathcal{S},\bx)] = \Delta_c(\mathcal{S}). \label{eq:sum-of-deltas}
\end{equation} 









At the heart of our proofs of both~\Cref{thm:warmup,thm:main} is a powerful inequality from the analysis of boolean functions, due to O'Donnell, Saks, Schramm, and Servedio~\cite{OSSS05}.  We will need the following variant of the OSSS inequality, which is due to Jain and Zhang~\cite{JZ11}: 

\begin{theorem}[OSSS inequality;~\cite{JZ11} version]
\label{thm:OSSS}
  For all functions $f : \zo^n \to \bits$ and strategies $\mathcal{S}$, 
  \[ \bias(f) -\error_f(\mathcal{S}) \le \sum_{i=1}^n \delta_i(\mathcal{S})\cdot \Inf_i(f). \] 
\end{theorem}

As a straightforward consequence of~\Cref{thm:OSSS}, we show that functions $f$ that admit low-error strategies with low expected cost have a variable with large {\sl cost-weighted influence}: 
  
\begin{lemma}
\label{lem: cost-OSSS}
For all functions $f : \zo^n \to \bits$, cost vectors $c \in \N^n$, and strategies $\mathcal{S}$, there is a coordinate $i\in [n]$ such that 
\[ \frac{\Inf_i(f)}{c_i} \ge \frac{\bias(f)-\error_f(\mathcal{S})}{\Delta_c(\mathcal{S})}. \] 
\end{lemma}

\begin{proof}
We have that
\begin{align*}
\bias(f) - \error_f(\mathcal{S}) &\le \sum_{i=1}^n \delta_i(\mathcal{S})\cdot \Inf_i(f) \tag*{(\Cref{thm:OSSS})} \\
&= \sum_{i=1}^n \delta_i(\mathcal{S})\cdot c_i \cdot \frac{\Inf_i(f)}{c_i} \\
&\le \max_{i \in [n]} \bigg\{  \frac{\Inf_i(f)}{c_i} \bigg\} \cdot \sum_{i=1}^n \delta_i(\mathcal{S}) \cdot c_i \\ 
&= \max_{i \in [n]} \bigg\{  \frac{\Inf_i(f)}{c_i} \bigg\} \cdot \Delta_c(\mathcal{S}) \tag*{(Identity~(\ref{eq:sum-of-deltas}))}, 
\end{align*} 
and the lemma follows. 
\end{proof}

For the proof of~\Cref{thm:warmup}, we will use the following special case of~\Cref{lem: cost-OSSS}: 

\begin{corollary}
\label{cor:zero-error-OSSS}
For all functions $f : \zo^n \to \bits$ and cost vectors $c\in \N^n$, there is a coordinate $i\in [n]$ such that 
\[ \frac{\Inf_i(f)}{c_i} \ge \frac{\bias(f)}{\opt'}, \quad\text{where $\opt'\coloneqq \opt(f,c,0)$.} \] 
\end{corollary} 

\begin{proof} 
This follows by applying~\Cref{lem: cost-OSSS} with $\mathcal{S}$ being the zero-error strategy for $f$ that achieves $\Delta_c(\mathcal{S}) = \opt'$. 
\end{proof} 

We will analyze a slight variant (and simplification) of our query strategy $\mathcal{S^\star}$ described in~\Cref{fig:our-strategy}: 

\begin{figure}[H]
  \captionsetup{width=.9\linewidth}
\begin{tcolorbox}[colback = white,arc=1mm, boxrule=0.25mm]
\vspace{3pt} 

Query strategy $\mathcal{S}^\diamond(f,c,\eps)$ for an unknown input $\underline{x}$:  \vspace{6pt} 
 
\ \ \ \ \ \  while $\bias(f) > \eps$:   \vspace{4pt} 

\ \ \ \ \ \ \ \ \ \ Let $i \in [n]$ be the coordinate with the largest cost-weighted influence:
\[ \frac{\Inf_i(f)}{c_i} \ge \frac{\Inf_j(f)}{c_j} \quad \text{for all $j\in [n]$}.\] 
\ \ \ \ \ \ \ \ \ \ Query the $i$-th coordinate of $\underline{x}$, and restrict $f$ accordingly: $f = f_{x_i = \underline{x}_i}$  \vspace{6pt}

\ \ \ \ \ \ Output $\sign(\E[f])$.
\end{tcolorbox}
\caption{A simplified version of our strategy $\mathcal{S}^\star$ described in~\Cref{fig:our-strategy}.}
\label{fig:our-strategy-2}
\end{figure}

Note that unlike $\mathcal{S}^\star$, this strategy $\mathcal{S}^\diamond$ does not keep a counter, and does not need to ``know" the values of $\opt$ and $\Inf(f)$.  

The fact that $\mathcal{S}^\diamond$ is an efficient randomized strategy follows from the fact that $\bias(f)$ and variable influences $\Inf_i(f)$ can be estimated to high accuracy (with high probability) given blackbox query access to $f$.  These are standard random sampling arguments and we omit the details (see e.g.~Section 8 of~\cite{BLT20} where the details are spelt out).  To prove~\Cref{thm:warmup}, it remains to show that $\mathcal{S}^\diamond$ is an $\eps$-error strategy, and to bound its expected cost.   

\begin{proposition}[$\mathcal{S}^\diamond$ is an $\eps$-error strategy]
$\Pr[\mathcal{S}^\diamond(\bx) \ne f(\bx)] \le \eps$. 
\end{proposition}

\begin{proof}
This follows directly from the termination condition of our algorithm ($\bias(f) \le \eps$), and the fact that $\Pr[f(\bx)\ne \sign(\E[f])] = \bias(f)$. 
\end{proof}








\begin{lemma}[Bounding the expected cost of $\mathcal{S}^\diamond$]
\label{lem:cost-weighted-depth} 
$\Delta_c(\mathcal{S}^\diamond) \le \purple{\overline{\opt}'} \cdot \Inf(f)/\eps$. 
\end{lemma}

\begin{proof}



For each node $v$ in the decision tree representation of $\mathcal{S}^\diamond$, we write $\pi_v$ to denote the root-to-$v$ path in $\mathcal{S}^\diamond$, and $x_{i(v)}$ to denote the variable that is queried at $v$. Since $\mathcal{S}^\diamond$ is constructed using the strategy described in~\Cref{fig:our-strategy-2}, we have that $x_{i(v)}$ is the variable that maximizes {\sl cost-weighted influence} in the restriction of $f$ by $\pi_v$ (written $f_{\pi_v}$). By~\Cref{cor:zero-error-OSSS}, it follows that for every $v \in \mathcal{S}^\diamond$, 
\begin{equation}   \frac{\Inf_{i(v)}(f_{\pi_v})}{c_{i(v)}} \ge \frac{\bias(f_{\pi_v})}{\purple{\opt(f_{\pi_v},c,0)}},\label{eq:large-cost-weighted-inf}
\end{equation} 
where we have used the fact that $\opt(f_{\pi_v},c,0)\le \purple{\overline{\opt}(f_{\pi_v},c,0) \leq \overline{\opt}(f,c,0) = \overline{\opt}'}$.

Due to~\Cref{fact:influence split} and~\Cref{cor:average influence}, we have that 
\begin{align*}
    \Inf(f) - \AvgInf_f(\mathcal{S}^\diamond) &= \sum_{v \in \mathcal{S}^\diamond} 2^{-|\pi_v|} \cdot \Inf_{i(v)}(f_{\pi_v}) \tag*{(\Cref{cor:average influence})} \\
    &\ge \sum_{v\in \mathcal{S}^\diamond} 2^{-|\pi_v|} \cdot \frac{\bias(f_\pi) \cdot c_{i(v)}}{\purple{\opt(f_{\pi_v}, c, 0)}} \tag*{(\Cref{eq:large-cost-weighted-inf})} \\
    &\ge \frac{\eps}{\purple{\overline{\opt'}}} \sum_{v\in \mathcal{S}^\diamond} 2^{-|\pi_v|}\cdot c_{i(v)} \tag*{(Termination condition of $\mathcal{S}^\diamond$)} \\
    &= \frac{\eps}{\purple{\overline{\opt'}}} \cdot \Delta_c(\mathcal{S}^\diamond) 
\end{align*}

Since $\AvgInf_f(\mathcal{S}^\diamond) \ge 0$, we conclude that 
\[  \Delta_c(\mathcal{S}^\diamond) \le \Inf(f) \cdot \frac{\purple{\overline{\opt'}}}{\eps}, \] 
and the proof is complete.
\end{proof} 




\section{Proof of~\Cref{thm:main}} 

\begin{reptheorem}{thm:main}
    Let $f : \zo^n \to \{\pm 1 \}$ be a function, $c \in \N^n$ be a cost vector, and $\eps \in (0,\frac1{2})$ be an error parameter.  There is an efficient $O(\eps)$-error randomized strategy $\mathcal{S}^\star$ for $f$ with expected cost $\Delta_c(\mathcal{S}^\star) \le \opt\cdot \Inf(f)/\purple{\eps^2}$.
\end{reptheorem}

\purple{ 
\Cref{thm:main} is a straightforward consequence of the following, which bounds the expected cost of $S^\star$ in terms of $\overline{\opt}$ instead of $\opt$: 

\begin{theorem} 
\label{thm:main-worst-case} 
    Let $f : \zo^n \to \{\pm 1 \}$ be a function, $c \in \N^n$ be a cost vector, and $\eps \in (0,\frac1{2})$ be an error parameter.  There is an efficient $O(\eps)$-error randomized strategy $\mathcal{S}^\star$ for $f$ with expected cost $\Delta_c(\mathcal{S}^\star) \le \overline{\opt}\cdot \Inf(f)/\eps$.
\end{theorem}

\begin{proof}[Proof of~\Cref{thm:main} assuming~\Cref{thm:main-worst-case}]
     Let $\mathcal{S}_{\opt}$ be the $\eps$-error strategy with expected cost $\opt$. Consider the strategy $\overline{\mathcal{S}_{\opt}}$ which executes $\mathcal{S}_{\opt}$, except that if it would make a query leading to a cost greater than $\frac{\opt}{\eps}$, it instead returns $1$ (arbitrarily). By Markov's inequality, $\overline{\mathcal{S}_{\opt}}$ differs from $\mathcal{S}_{\opt}$ with probability at most $\eps$, and therefore $\overline{\mathcal{S}_{\opt}}$ is a $2\eps$-error strategy for $f$. Furthermore, $\overline{\mathcal{S}_{\opt}}$ has worst-case cost at most $\frac{\opt}{\eps}$. \Cref{thm:main} holds by applying \Cref{thm:main-worst-case} with $\overline{\opt} \leq \frac{\opt}{\eps}$.
\end{proof} 


}

To prove~\purple{\Cref{thm:main-worst-case}}, we consider the query strategy shown in \Cref{fig:our-strategy-budget}, which is defined for any budget $B > 0$.  First we show that if $B$ is set to $\purple{\overline{\opt}} \cdot \Inf(f) / \eps$, then $\mathcal{S}_B^\star$ satisfies \Cref{thm:main}. (Note that for \purple{appropriately chosen} $B$, $\mathcal{S}_B^\star$ is exactly the strategy $\mathcal{S}^\star$ described in~\Cref{fig:our-strategy} in the introduction.) Then, in \Cref{sec:stop} we give an algorithm for computing this value of the budget $B$.

\begin{figure}[h]
    \captionsetup{width=.9\linewidth}
    \begin{tcolorbox}[colback = white,arc=1mm, boxrule=0.25mm]
        \vspace{3pt} 
        
        Query strategy $\mathcal{S}_B^\star(f,c,\eps)$ for an unknown input $\underline{x}$:  \vspace{6pt} 
        
        \ \ Initialize {\tt counter} $\coloneqq 0$\vspace{4pt}  
        
        \ \ while {\tt counter} $ < B$:  \vspace{4pt} 
        
        
        \ \ \ \ \ \ Let $i \in [n]$ be the coordinate with the largest cost-weighted influence:
        \[ \frac{\Inf_i(f)}{c_i} \ge \frac{\Inf_j(f)}{c_j} \quad \text{for all $j\in [n]$}.\] 
        \ \ \ \ \ \ if ${\tt counter} + c_i \geq B$:\vspace{4pt} 
        
        \ \ \ \ \ \ \ \ \ \ \ \ With probability $\frac{B - {\tt counter}}{c_i}$:\vspace{4pt}
        
        \ \ \ \ \ \ \ \ \ \ \ \ \ \ \ \ \ \ Query the $i$-th coordinate of $\underline{x}$ and output $\sign(\E[f_{x_i = \underline{x}_i}])$. \vspace{4pt}

        \ \ \ \ \ \ \ \ \ \ \ \ With probability $1 - \frac{B - {\tt counter}}{c_i}$:\vspace{4pt}
        
        \ \ \ \ \ \ \ \ \ \ \ \ \ \ \ \ \ \ Output $\sign(\E[f])$. \vspace{4pt}
        
        \ \ \ \ \ \ Query the $i$-th coordinate of $\underline{x}$, and restrict $f$ accordingly: $f = f_{x_i = \underline{x}_i}$  \vspace{4pt}
        
        \ \ \ \ \ \ {\tt counter+=} $c_i$  
        
        \end{tcolorbox}
    \caption{The query strategy $\mathcal{S}_B^\star$ is defined for any budget $B \geq 0$. It takes as input, a representation of $f : \zo^n\to \{\pm 1\}$, a cost vector $c\in \N^n$, and an error parameter $\eps$.  It has query access to an unknown input~$\underline{x}$, and its goal is to compute $f(\underline{x})$ in a computationally- and cost-efficient manner.}
    \label{fig:our-strategy-budget}
\end{figure}

\begin{lemma}[Expected cost of $\mathcal{S}_B^\star$]
    For any $B > 0$, the strategy $\mathcal{S}_B^\star(f,c,\eps)$ has expected cost exactly equal to
    \begin{align*}
        \Delta_c(\mathcal{S}_B^\star) = \min(B, \sum_{i=1}^n c_i).
    \end{align*}
\end{lemma}
\begin{proof}
    On any fixed input $x$, the expected cost of $\mathcal{S}_B^\star$ is $\min(B, \sum_{i=1}^n c_i)$, where $\sum_{i=1}^n c_i$ is the cost of a strategy that makes \textit{all} possible queries. Therefore, the expected cost over all $x$ is also $\min(B, \sum_{i=1}^n c_i)$.
\end{proof}

\begin{lemma}[Accuracy of $\mathcal{S}_B^\star$]
    \label{lem:accuracy main}
    Suppose that $\mathcal{S}_B^\star$ computes influences and expectations exactly. Then, for any $B^\star > \purple{\overline{\opt}} \cdot \Inf(f) / \eps$, the strategy $\mathcal{S}_{B^\star}^\star$ is a less than $2\eps$-error randomized strategy for $f$.
\end{lemma}
\begin{proof}
    We prove the contrapositive. Suppose that $\mathcal{S}_{B^\star}^\star$ has error at least $2\eps$. We will prove $B^\star \leq \purple{\overline{\opt}} \cdot \Inf(f) / \eps$.
    
    We first prove that for any $0 \leq B \leq B^\star$,
    \begin{align}
        \label{eq:Influence decreases}
        \AvgInf_f(\mathcal{S}_{B}^\star) \leq \Inf(f) - \frac{B \cdot \eps}{\purple{\overline{\opt}}}.
    \end{align}
     Recall \Cref{cor:average influence},
    \begin{align*}
        \AvgInf_f(\mathcal{S}_{B}^\star) = \Inf(f) - \Ex_{\bx \in \zo^n}\left[\sum_{\alpha \in \paths(\mathcal{S}_{B}^\star, \bx)} \Inf_{\query(\mathcal{S}_{B}^\star, \alpha)}(f_\alpha) \right]
    \end{align*}
    Since $\mathcal{S}_{0}^\star$ makes no queries, $\AvgInf_f(\mathcal{S}_{B}^\star) = \Inf(f)$. Therefore, in order to prove \Cref{eq:Influence decreases}, we just show the following bound on the derivative of $\AvgInf$ for any $B \leq B^\star$.
    \begin{align*}
        -\frac{d}{dB}\left(\AvgInf_f(\mathcal{S}_{B}^\star)\right) \geq \frac{\eps}{\purple{\overline{\opt}}}
    \end{align*}
    For any $B \leq B^\star$  there is some sufficiently small $\eps$ for which the strategies $\mathcal{S}^\star_{B}$ and $\mathcal{S}_{B + \eps}$ are identical, except for each $x \in \zo^n$,
    \begin{align*}
        \ell_B(x) \coloneqq \leaf(\mathcal{S}_B^\star, x) \quad \text{and} \quad
        \nexti_B(x) \coloneqq \query(\mathcal{S}_{B + \eps}^\star,\ell_B(x)),
    \end{align*}
     $\mathcal{S}_{B + \eps}$ has an additional $\frac{\eps}{c_{\nexti_B(x)}}$ probability of querying $x_{\nexti_B(x)}$ as its last query. Combining the relation between $\mathcal{S}^\star_{B}$ and $\mathcal{S}_{B + \eps}$ with \Cref{cor:average influence},
    \begin{align*}
        -\frac{d}{dB}\left(\AvgInf_f(\mathcal{S}_{B}^\star)\right) &=
         \frac{d}{dB}\left( \Ex_{\bx \in \zo^n}\left[\sum_{\alpha \in \paths(\mathcal{S}_{B}^\star, \bx)} \Inf_{\query(\mathcal{S}_{B}^\star, \alpha)}(f_\alpha) \right]\right)\\
         &= \Ex_{\bx \in \zo^n}\left[\lim_{\eps \to 0}\frac{(\eps/ c_{\nexti_B(\bx)}) \cdot \Inf_{\nexti_B(\bx)}(f_{\ell_B(\bx)}) }{\eps} \right] \\
         &=  \Ex_{\bx \in \zo^n}\left[\frac{\Inf_{\nexti_B(\bx)}(f_{\ell_B(\bx)})}{c_{\nexti_B(\bx)}}\right] 
    \end{align*}
    
    Let $\purple{\overline{\mathcal{S}_\opt}}$ be the $\eps$-error strategy for $f$ with \purple{worst-case} cost $\purple{\overline{\opt}}$.
    \begin{align*}
        \Ex_{\bx \in \zo^n}\left[\frac{\Inf_{\nexti_B(\bx)}(f_{\leaf(\mathcal{S}_B, \bx)})}{c_{\nexti_B(\bx)}}\right] &\geq \Ex_{\bell = \leaf_B(\bx)}\left[\frac{\bias(f_{\bell}) - \error_{f_{\bell}}((\purple{\overline{\mathcal{S}_\opt}})_{\bell})} 
        {\Delta_c((\purple{\overline{\mathcal{S}_\opt}})_{\bell})}\right] \tag{\Cref{lem: cost-OSSS}} \\
        &\geq \Ex_{\bell = \leaf_B(\bx)}\left[\frac{\bias(f_{\bell}) - \error_{f_{\bell}}((\purple{\overline{\mathcal{S}_\opt}})_{\bell})} 
        {\purple{\overline{\opt}}}\right] \tag{$\Delta_c((\purple{\overline{\mathcal{S}_\opt}})_{\bell}) \leq \purple{\overline{\Delta}_c((\purple{\overline{\mathcal{S}_\opt}})_{\bell}) \leq\overline{\Delta}_c(\purple{\overline{\mathcal{S}_\opt}}) = \purple{\overline{\opt}}}$} \\
        &= \frac{1}{\purple{\overline{\opt}}} \left(\Ex_{\bell = \leaf_B(\bx)}\left[ \bias(f_{\bell}) \right] - \Ex_{\bell = \leaf_B(\bx)}\left[\error_{f_{\bell}}((\purple{\overline{\mathcal{S}_\opt}})_{\bell}) \right] \right)
    \end{align*}
    Since $\mathcal{S}^\star_B$ is $f$-consistent, the quantity $\Ex_{\bell = \leaf_B(\bx)}\left[ \bias(f_{\bell}) \right]$ is just $\error_f(\mathcal{S}_B^\star)$. Since $B < B^\star$, this is at least the error of $\mathcal{S}^\star_{B^\star}$ which we assumed is at least $2\eps$. Similarly, $\Ex_{\bell = \leaf_B(\bx)}\left[\error_{f_{\bell}}((\purple{\overline{\mathcal{S}_\opt}})_{\bell}) \right]$ is $\error_f(\purple{\overline{\mathcal{S}_\opt}})$, which is at most $\eps$. Therefore,
    \begin{align*}
        \Ex_{\bx \in \zo^n}\left[\frac{\Inf_{\nexti_B(\bx)}(f_{\leaf(\mathcal{S}_B, \bx)})}{c_{\nexti_B(\bx)}}\right] \geq \frac{1}{\purple{\overline{\opt}}} \left(2\eps - \eps \right) = \frac{\eps}{\purple{\overline{\opt}}}.
    \end{align*}
    We conclude that \Cref{eq:Influence decreases} holds for $0 \leq B \leq B^\star$. Plugging in $B = B^\star$ and using the fact that for any incomplete strategy, including $\mathcal{S}_{B}^\star$, average influence is nonnegative, we have that
    \begin{align*}
        0 &\leq \Inf(f) - \frac{B^\star \cdot \eps}{\purple{\overline{\opt}}} \\
        B^\star &\leq \purple{\overline{\opt}} \cdot \Inf(f) / \eps. \qedhere 
    \end{align*}
\end{proof}

\begin{remark}[Handling errors in computing expectations]
    Our analysis in \Cref{lem:accuracy main} assumes that influences and expectations can be computed exactly. When running $\mathcal{S}_B^\star$ in practice, those values would be estimated through sampling. It is not hard to use standard Chernoff bounds and union bound to ensure that the probability any influence is off by more than a multiplicative factor of $2$ or expectation off by an additive $\eps$ is tiny. This would give an $O(\eps)$ rather than $2\eps$ error bound, which is consistent with \Cref{thm:main}.
\end{remark}

As long as influences and expectations are estimated using polynomially many random samples, $\mathcal{S}_B^\star$ is clearly computational efficient. The last step in proving \Cref{thm:main} is showing how to find a suitable value for $B$.

\subsection{An algorithm for computing the budget}
\label{sec:stop} 
If $\mathcal{S}^\star_B$ knows what $\purple{\overline{\opt}}$, it can just estimate $\Inf(f)$ by summing up estimates of the individual coordinate influences and set $B = \purple{\overline{\opt}}\cdot \Inf(f) / \eps$. Otherwise, we give a strategy for deciding which $B$ to use satisfying $B \leq \purple{\overline{\opt}} \cdot \Inf(f) / \eps$ and $\mathcal{S}_B^\star$ being an $O(\eps)$-error randomized strategy.

For any fixed $B$, we can efficiently estimate the error of $\mathcal{S}_B^\star$ with random samples using the definition.
\begin{align*}
    \error_f(\mathcal{S}_B^\star) \coloneqq \Prx_{\bx \sim \bits^n}[\mathcal{S}_B^\star(\bx) \neq f(\bx)]
\end{align*}
Our algorithm simply computes $\mathcal{S}_B^\star(\bx)$ for many random $\bx \sim \bits^n$. From this, it can determine, with high probability, whether $\mathcal{S}_B^\star$ is an $O(\eps)$-error strategy. Our algorithm just tries the budgets $B = 2^i \cdot \min_{j \in [n]} c_{j}$  for $i = 1,2,\ldots$ until it finds a budget on which $\error_f(\mathcal{S}_B^\star) < \eps$. Then it uses the corresponding strategy. 

Note that estimating a good choice for $B$ does not require any queries to the particular input given, only to random inputs. For any particular input $\underline{x}$, we can compute $\mathcal{S}_B(\underline{x})$ after determining $B$ using the above algorithm with an average cost of only $B$.

\section{Lower bounds and optimality of our analyses} 
\subsection{Proof of~\Cref{claim:optimality}}

\begin{repclaim}{claim:optimality}
Consider the unit-cost setting, $c = 1^n$.  For any $\eps \in (0, \frac{1}{4})$, there is a function $f : \zo^n \to \{ \pm 1\}$ such that our strategy $\mathcal{S}^\star$ has to incur expected cost $\opt(f, c, \eps) \cdot \tilde{\Omega}_\eps(\Inf(f))$ in order to achieve error $\eps$. 
\end{repclaim} 
\begin{proof}
Our proof will apply to both $\mathcal{S}^\star$ and $\mathcal{S}^\diamond$. These algorithms always query the coordinate with the largest cost-weighted influence. Our lower bound will apply to any such strategy, regardless of the \emph{stopping conditions} -- i.e., how the strategy determines when to terminate for a particular input $x$. Since $c=1^n$, such strategies always query the most influential variable of $f$. \Cref{claim:optimality} is a simple extension of the following theorem from \cite{BLT20}.
\begin{lemma}[Theorem 4(b) from \cite{BLT20}]
    \label{lem:BLT lower bound}
    For every $\varepsilon' \in (0, \frac{1}{2})$ and $s \in \N$, there is a function $f : \{0,1\}^n \rightarrow \{\pm 1\}$ exactly computable by a size $s$ decision tree with the following property. Let $\mathcal{S}$ be \emph{any} strategy that recursively queries the most influential variable, regardless of stopping conditions. If $\mathcal{S}$ has error at most $\varepsilon'$, then when it is represented as a decision tree, $\mathcal{S}$ has size at least $s^{\tilde{\Omega}_\varepsilon(\log s)}$.
\end{lemma}

Let $f$ be the function defined in \Cref{lem:BLT lower bound} for $\eps' = 2 \cdot \eps$. Suppose, for the sake of contradiction, that $\mathcal{S}^*$ has expected cost only $\Delta = \opt(f, c, \eps) \cdot \tilde{o}_\eps(\Inf(f))$. Then,
\begin{align*}
    \Delta &= \opt(f, c, \eps) \cdot \tilde{o}_\eps(\Inf(f)) \\
    &\leq \opt(f, c, 0) \cdot \tilde{o}_\eps(\Inf(f)) \\
    &\leq  \log(s) \cdot \tilde{o}_\eps(\log(s)) = \tilde{o}_\eps(\log^2 s)
\end{align*}
Consider the strategy $\mathcal{S}_{\text{cutoff}}$ which runs $\mathcal{S}^\star$ with one modification: If $\mathcal{S}^\star$ wishes to make a query that leads to a total cost of more than $\frac{\Delta}{\varepsilon}$, then $\mathcal{S}_{\text{cutoff}}$ returns $0$ instead of exceeding the cutoff $\frac{\Delta}{\varepsilon}$. By Markov's inequality, $\Pr[\mathcal{S}^\star(\bx) \ne  \mathcal{S}_\text{cutoff}(\bx)] \leq \eps$. Then

\begin{align*}
    \error_f(\mathcal{S}_{\text{cutoff}}) &\leq \error_f(\mathcal{S}^\star) + \Pr[\mathcal{S}^\star(\bx) \ne  \mathcal{S}_\text{cutoff}(\bx)] \\
    &\leq \eps + \eps = 2\eps.
\end{align*}
$\mathcal{S}_{\text{cutoff}}$ is also a strategy that recursively queries the most influential variable, just with different stopping conditions then $\mathcal{S}^\star$. By \Cref{lem:BLT lower bound}, when represented as a decision tree, $\error_f(\mathcal{S}_{\text{cutoff}})$ has size at least $s^{\tilde{\Omega}_\varepsilon(\log s)}$. However, it also has worse case depth $\frac{\Delta}{\eps}$, which means it has decision tree size at most
\begin{align*}
    2^{\frac{\Delta}{\eps}} = 2^{\frac{\tilde{o}_\eps(\log^2 s)}{\eps}} =2^{\tilde{o}_\eps(\log^2 s)} = s^{\tilde{o}_\eps(\log s)}.\\
\end{align*}
We have reached a contradiction which proves \Cref{claim:optimality}.\end{proof}

\subsection{Proof of~\Cref{claim:breakthrough}} 

\begin{repclaim}{claim:breakthrough}
Consider the unit-cost setting, $c = 1^n$.  Suppose that for all functions $f : \zo^n \to \{\pm 1\}$ that are computable by $\poly(n)$-size decision trees, there is an efficient $O(\eps)$-error strategy $\mathcal{S}$ for $f$ with expected cost $\opt \cdot o_\eps(\Inf(f))$.   Furthermore, suppose that $\mathcal{S}$ only requires blackbox query access to $f$.  Then there is an $n^{o(\log n)}$-time query algorithm for properly PAC learning polynomial-size decision trees under the uniform distribution.  
\end{repclaim} 

\begin{proof}
Let $f : \zo^n \to \{ \pm 1\}$ be computable by a size-$s$ decision tree $T$ where $s = \poly(n)$.   Consider the truncation $T_{\mathrm{trunc}}$ of $T$ to depth $\log(s/\eps)$, replacing all truncated branches with arbitrary leaf values.  Since
\[ \Pr[\,T_{\mathrm{trunc}}(\bx) \ne T(\bx)\,] \le \mathop{\sum_{\text{paths $\pi$ in $T$}}}_{|\pi| > \log(s/\eps)} \Pr[\,\bx \text{ follows $\pi$}\,] \le s \cdot 2^{-\log(s/\eps)} = \eps, \] 
this tree $T_{\mathrm{trunc}}$ witnesses the existence of an $\eps$-error strategy for $f$ with worst-case cost $\log(s/\eps)$; consequently, $\opt \le \log(s/\eps)$.   By the assumption of~\Cref{claim:breakthrough}, there is an efficient $O(\eps)$-error strategy $\mathcal{S}$ for $f$ with expected cost 
\begin{align*}
\E[\cost(\mathcal{S},\bx)] &\le \opt \cdot o_\eps(\Inf(f)) \\
&\le \log(s/\eps) \cdot o_\eps(\Inf(f)) \\
 &\le \log(s/\eps) \cdot o_\eps(\log s) \\
 &= o_\eps((\log s)^2) \eqqcolon \Delta.
 \end{align*}
The final inequality above uses the bound $\Inf(f)\le \log s$; this is a standard (and easy to prove) fact; see e.g.~\cite{OS07}).  By Markov's inequality, we have that: 
\[ \Pr\bigg[\cost(\mathcal{S},\bx) > \frac{\Delta}{\eps}\bigg] \le \eps. \] 
In other words, the decision tree representation of $\mathcal{S}$, when truncated to depth $\Delta/\eps$ (recall that we are in the unit-cost setting), is $\eps$-close to $f$.  Since $\mathcal{S}$ is an efficient strategy, this truncated tree can be constructed in time polynomial in its size, 
\[ 2^{\Delta/\eps} = 2^{o_\eps((\log s)^2)/\eps} = n^{o_\eps(\log n)}\]
using only blackbox query access to $f$.  This completes the proof of~\Cref{claim:breakthrough}. 
\end{proof} 

 \section*{Acknowledgements} 

We thank Ryan O'Donnell for telling us about~\cite{CFGKRS00}, John Wright for an enjoyable discussion, and the SODA reviewers for their thoughtful and valuable feedback.

The authors were supported by NSF award 192179 and NSF CAREER award
1942123.


\bibliography{most-influential}{}
\bibliographystyle{alpha}

\end{document}